\documentclass[preprint,12pt]{elsarticle}

\usepackage{amssymb,amsmath}
\usepackage{multirow}
\usepackage{algorithm}
\usepackage{algorithmic}
\usepackage{color}
\usepackage{graphicx}

\newtheorem{theorem}{Theorem}
\newtheorem{lemma}{Lemma}

\newdefinition{rmk}{Remark}
\newproof{proof}{Proof}
\newcommand{\be}{\begin{equation}}
\newcommand{\ee}{\end{equation}}
\def\bea{\begin{eqnarray}}
\def\eea{\end{eqnarray}}
\def\nn{\nonumber}

\newcommand{\bbmat}{\left\{ \begin{array}}
\newcommand{\ebmat}{\end{array} \right. }

\def\endzm{\hfill $\Box$ \bigskip}

\def\b0{\boldsymbol{\rm 0}}
\def\bo{\boldsymbol{o}}

\def\bs{\boldsymbol{s}}

\def\bx{\boldsymbol{x}}

\def\by{\boldsymbol{y}}
\def\bq{\boldsymbol{q}}

\def\bI{\boldsymbol{I}}

\def\tM{\textbf{M$k$K}}

\begin{document}

\begin{frontmatter}



\title{Approximation algorithms for $k$-submodular maximization subject to a knapsack constraint}

\author[1]{Hao Xiao}
\author[1]{Qian Liu}
\author[1]{Yang Zhou}
\author[1]{Min Li\corref{cor1}}

\cortext[cor1]{Email: liminemily@sdnu.edu.cn}

\address[1]{School of Mathematics and Statistics, Shandong Normal University, P.R. China.}

\begin{abstract}
In this paper, we study the problem of maximizing $k$-submodular functions subject to a knapsack constraint.
For monotone objective functions, we present a $\frac{1}{2}(1-e^{-2}) \approx 0.432$ greedy approximation algorithm.
For the non-monotone case, we are the first to consider the knapsack problem and provide a greedy-type combinatorial algorithm with approximation ratio $\frac{1}{3}(1-e^{-3}) \approx 0.317$.
\end{abstract}






\end{frontmatter}







\section{Introduction}
For a positive integer $k$, denote $[k]=\{1,2,\ldots,k\}$. Define the following set composing of the $k$-tuples of disjoint subsets in a finite set $E$:
\[
(k+1)^E:=\{(X_1,X_2,\ldots,X_k) | X_i\subseteq E, \forall i\in [k]; X_i\cap X_j=\varnothing, \forall i,j\in [k], i\neq j\}.
\]
A function $f:(k+1)^E\rightarrow \mathbb{R}_+$ is \emph{$k$-submodular} if for any $\bx=(X_1,\dots, X_k)\in (k+1)^E$ and $\by=(Y_1,\dots, Y_k)\in (k+1)^E$, we have
\[ f(\bx)+f(\by) \ \ge \ f(\bx\sqcup\by)+f(\bx\sqcap\by), \]
where
\begin{eqnarray*}
\bx\sqcup\by&=&\bigg(X_1\cup Y_1\setminus(\bigcup_{i\ne 1}(X_i\cup Y_i)), \dots,X_k\cup Y_k\setminus(\bigcup_{i\ne k}(X_i\cup Y_i))\bigg);\\
\bx\sqcap\by&=&\bigg(X_1\cap Y_1, \dots, X_k\cap Y_k\bigg).
\end{eqnarray*}
Evidently,  a $1$-submodular function is just the standard submodular function and hence $k$-submodularity generalizes submodularity.


We introduce a partial order $\preceq$ on $(k+1)^E$: $\forall \bx, \by\in (k+1)^E$, $\bx \preceq \by$ if $X_i\subseteq Y_i$, for each $i\in [k]$. Then a function $f$ is \emph{monotone} if for any $\bx \preceq \by$, we have
$f(\bx)\leq f(\by)$. Moreover, we use $P(\bx)=\cup_{i\in [k]}X_i$ to denote the support set of $\bx$.

As a generalization of submodular function, $k$-submodular functions have many applications~\cite{EN2022,AV,TN}.
For example, in the sensor placement problem, assuming that there are $k$ different types of sensors as candidates, the location of the sensors as well as their types should be considered when we make decisions.
This problem can then be formulated as a $k$-submodular maximization model.
Moreover, if there is budget limitation, this corresponds to a knapsack constraint.

\textbf{Our contributions.}
In this paper, we present approximation algorithms for the problem of maximizing a $k$-submodular function under a knapsack constraint.
For each element $e\in E$, there is a nonnegative integer $c_e$ denoting its cost; and for each subset $S\subseteq E$, let $c(S)=\sum_{e\in S}c_e$ denote the total cost of $S$.
Given an integer budget $L$, the problem of maximizing a $k$-submodular function $f$ under a knapsack constraint (denoted by \tM) is:
\[
\max_{\bx\in(k+1)^E} f(\bx) \quad {\rm s. t.} \ c(P(\bx))\le L.
\]

When the objective function $f$ is monotone, a deterministic  $\frac{1}{2}(1-e^{-1})$-approximation algorithm was proposed~\cite{TWC} and a randomized algorithm with approximation ratio  $(\frac{1}{2}-\varepsilon)$ was designed by introducing a continuous greedy technique~\cite{WZ2017}. In this monotone case,
we improve the approximation ratio to $\frac{1}{2}(1-e^{-2})$ based on the greedy technique (this technique has been used often in the maximization of submodular and $k$-submodular functions, see, e.g., ~\cite{M,TWC}).
When the function $f$ is non-monotone, we firstly prove that the greedy-based algorithm is $\frac{1}{3}(1-e^{-3})$-approximation.

\textbf{Related work.}
It is NP-hard to maximize a monotone $k$-submodular function even without constraint, for which case there exists a deterministic performance guarantee $\frac{1}{2}$~\cite{JS} as well as a randomized approximation ratio $\frac{k}{2k-1}$~\cite{ITY}. For maximizing a non-monotone objective function without constraint, there exists a deterministic algorithm and a randomized algorithm with approximation ratios $\frac{1}{3}$ and $\max\{\frac{1}{3}, \frac{1}{1+a}\}$ respectively, where $a=\max\left\{1, \sqrt{\frac{k-1}{4}}\right\}$~\cite{JS}. Recently, the randomized result was improved to $\frac{1}{2}$~\cite{ITY} and further to $\frac{k^2+1}{2k^2+1}$~\cite{H}.

For constrained $k$-submodular function maximization problem,  two types of size constraints have been investigated in the literature.
The first is the cardinality constraint where a solution is feasible if the number of elements in its support set is bounded by some given value $B$ from above (denoted by $|P(\bx)|\leq B$). When the objective function is monotone with a cardinality constraint, there exists a $\frac{1}{2}$-approximation algorithm~\cite{NY}. For the non-monotone case, there is an approximation ratio $\frac{1}{3}$~\cite{NT}.
The second size constraint studied is the individual size, where there are $k$ cardinality constraints on the output solution $\bx=(X_1,X_2,\ldots,X_k)$, i.e., $|X_i|\leq B_i$ for some given values $B_i$, $\forall i\in [k]$. For a monotone $k$-submodular function, there is a $\frac{1}{3}$-approximation algorithm~\cite{NY} under this constraint type.

For maximizing a monotone $k$-submodular function under a matroid constraint, there is a $\frac{1}{2}$-approximation algorithm~\cite{RY2020,S}. When the objective function is non-monotone, a $\frac{1}{3}$-approximation algorithm is presented~\cite{SLL2022}.
More results on $k$-submodular functions can be found
in~\cite{CQDe,SGW,XLZL2022,YLZL2022,YLZL2023a,YLZL2023b,ZCLL}.

\section{Preliminaries}\label{pre}
For any $e\in E$ and $\bx=(X_1,\dots, X_k)\in (k+1)^E$,
let $\bx_e$ denote the index of the subset in $\bx$ containing $e$ and set $\bx_e=0$ if $e$ does not belong to any subset of $\bx$. That is,
\[
\bx_e=\bbmat{rll}
&i, \;\;& {\rm if}\;  e \in X_i, \\ [1ex]
&0, \; \;& {\rm if}\;  e \notin P(\bx).
\ebmat
\]

The special $k$-tuple $\bI_{[e, i]}=\bigg(\varnothing, \ldots, \underbrace{\{e\}}\limits_{i{\rm-th}}, \dots,\varnothing\bigg)\in (k+1)^E$ contains only one item $e$ in its $i$-th subset.
Taking $e\in E\setminus P(\bx)$ and $i\in [k]$, we use $\Delta_{e, i}f(\bx)=f(\bx \sqcup\bI_{[e,i]})-f(\bx)$ to denote the marginal benefit or gain of adding $e$ to the $i$-th subset of $\bx$. $f$ is \emph{orthant submodular} if it satisfies diminishing marginal benefit, i.e., $\bx\preceq\by$ implies $\Delta_{e, i}f(\bx)\ge\Delta_{e, i}f(\by)$ for any  $e\in E\setminus P(\by)$ and $i\in [k]$. 
The $k$-submodular function $f$ is \emph{pairwise monotone} if $\Delta_{e, i}f(\bx)+\Delta_{e, j}f(\bx)\ge0$ for any $\bx\in(k+1)^E$, $e\in E\setminus P(\bx)$, and distinct $i, j \in [k]$.

The following lemmas are some preliminary results about $k$-submodular functions.
\begin{lemma}\label{mon-ksub}
(\cite{JS}) A function is $k$-submodular if and only if it is orthant submodular and pairwise monotone.
\end{lemma}

\begin{lemma}\label{lem-ksub}
(\cite{TWC})
Let $f : (k+1)^E \rightarrow \mathbb{R}_+$ be a $k$-submodular function. For any $\bx,\by\in (k+1)^E$ with $\bx\preceq \by$, we have
\[
f(\by)-f(\bx)\leq \sum_{e\in P(\by)\setminus P(\bx)}\Delta_{e,\by_e}f(\bx).
\]
\end{lemma}

%

In this paper, the following lemma also plays an important role.
\begin{lemma}\label{lemA111}(\cite{L})
If $A$ and $B$ are  arbitrary positive integers, $\rho_1>0$, and
$\rho_2,\ldots,\rho_A$ are all arbitrary nonnegative real numbers, then
\[
\frac{\sum_{i=1}^{A}\rho_i}{\min_{s=1, \dots, A}(\sum_{i=1}^{s-1}\rho_i+B\rho_s)}\ge 1-\bigg( 1-\frac{1}{B}\bigg) ^A\geq 1-e^{-A/B}.
\]
\end{lemma}

\section{A deterministic algorithm for \tM}\label{alg-ana}
In this section, we propose a deterministic approximation algorithm for \tM. The main idea of the algorithm is to enumerate all size-$w$ solutions and extend each of them greedily until the budget runs out. The greedy procedure in each iteration is to maximize the marginal density (i.e., the marginal gain divided by the element cost). The best solution is returned as outcome.
When the input function $f$ is monotone, this greedy method has been presented for maximizing submodular functions~\cite{M} and $k$-submodular functions~\cite{TWC}. Our main contribution is to provide an improved analysis of the approximation ratio from $\frac{1}{2}(1-e^{-1})$ to $\frac{1}{2}(1-e^{-2})$.
Moreover, we further show that this algorithm is $\frac13(1-e^{-3})$-approximation when $f$ is non-monotone.
There may exist some elements with ``big" cost and ``high" value so that their marginal density may be not the highest.
To prevent the algorithm from missing such elements, enumeration of $w$ elements is required before executing the greedy scheme, which can be calculated through analysis.
\begin{algorithm}[htb]
    \caption{Greedy algorithm for \tM }\label{greedy}
    \textbf{Input:} A set $E$ with cost $c$, a $k$-submodular function $f$, an integer $w$ (if $f$ is a monotone $k$-submodular function, $w = 4$; otherwise, $w=7$) and a budget $L$. \\
    \textbf{Output:} $\bs\in (k+1)^E$ with $c(P(\bs))\le L$.

    \begin{algorithmic}[1]
        \STATE $\bs\gets\arg\max f(\bs_0)$ subject to $|P(\bs_0)|=w-1$ and $c(P(\bs_0))\le L$
        \FOR{every $\bs^0 \ (|P(\bs^0)|=w$ and $c(P(\bs^0))\le L)$}
            \STATE $E^0:=E\setminus P(\bs^0)$ and $j:=1$
            \WHILE {$|E^{j-1}|\ne0$}
                \STATE $[e^j, i^j]\gets \arg\max_{e\in E^{j-1}, i\in[k]}\frac{\Delta_{e, i}f(\bs^{j-1})}{c_e}$
                \IF{$c(P(\bs^{j-1}))+c_{e^j}\le L$}
                    \STATE $\bs^j:=\bs^{j-1}\sqcup\bI_{[e^j, i^j]}$
                \ELSE
                    \STATE $\bs^j:=\bs^{j-1}$
                \ENDIF
                \STATE $E^j:=E^{j-1}\setminus\{e^j\}$ and $j:=j+1$
            \ENDWHILE
            \STATE $\bs\gets \arg\max\{f(\bs), f(\bs^j)\}$
        \ENDFOR
        \STATE \textbf{return } $\bs$
    \end{algorithmic}
\end{algorithm}

Let $\bo$ be an optimal solution with $|P(\bo)|=r$.
To analyze the approximation guarantee of the algorithm, we assume that $r\geq w$ because, if $r \le w-1$, the algorithm finds the optimal solution in the first step, where $w$ is the undetermined integer.
Then we order $P(\bo)=\{e_1,\ldots,e_{r}\}$ by defining a new sequence of $k$-tuples $\{\bq^j\}_{j=0}^{r}$ with $\bq^0=\varnothing$ and $P(\bq^{r})=P(\bo)$ according to the maximal gain:
\[
f(\bq^j)=\max_{e\in P(\bo)\setminus P(\bq^{j-1}),i\in [k]} f(\bq^{j-1} \sqcup \bI_{[e,i]}), \, j=1, \dots, r.
\]

Moreover, we can construct another sequence of $k$-tuples $\{\bar{\bo}^j\}_{j=0}^{r}$ with $\bar{\bo}^0=\bo$ and
$\bar{\bo}^j$ consists of the items in $P(\bo)$, where the index of items in $P(\bq^j)$ aligns with $\bq^j$, and the index of other items aligns with $\bo$, see, Figure~\ref{fig:Figure_1}.

That is, for $j = 1, \dots, r$, we define
\bea
\bar{\bo}^j &=& (\bar{\bo}^{j-1} \sqcup \bq^j) \sqcup \bq^j, \nn
\eea
then by using the unconstrained greedy analysis process in~\cite{JS}, we can get
\bea
f(\bar{\bo}^{j-1}) - f(\bar{\bo}^j) &\le& f(\bq^j) - f(\bq^{j-1})\label{without-constraint-iteration}\,\, (f \, {\rm is} \, {\rm monotone}),\\
f(\bar{\bo}^{j-1}) - f(\bar{\bo}^j) &\le& 2[f(\bq^j) - f(\bq^{j-1})]\label{without-constraint-iteration-non}\,\, ( f \, {\rm is} \, \text{non-monotone}).
\eea

\begin{figure}[htb]
    \centering
    \includegraphics[width=12cm,height=10cm,angle=0]{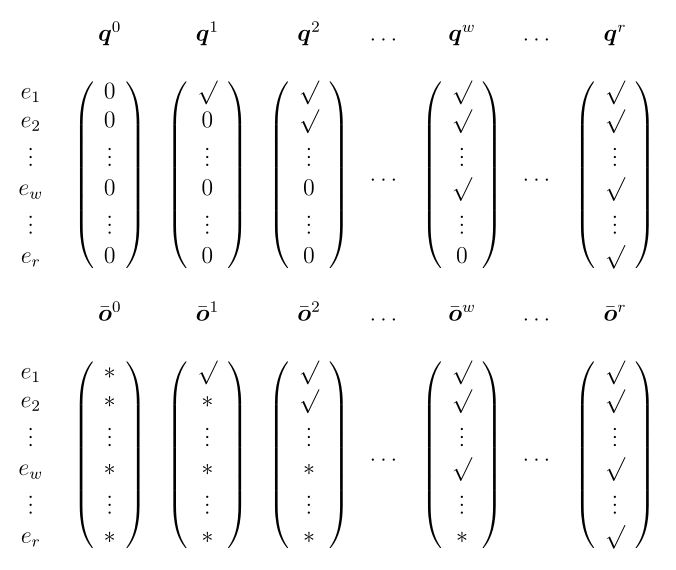}
    \caption{The constructing sequences of $\{\bq^j\}_{j=0}^{r}$ and $\{\bar{\bo}^j\}_{j=0}^{r}$.} \label{fig:Figure_1}
\end{figure}

By summing the first $w$-terms of inequalities (\ref{without-constraint-iteration}) and (\ref{without-constraint-iteration-non}), separatively, we naturally obtain
\bea
f(\bo) - f(\bq^w) &\le& f(\bar{\bo}^w) \label{without-constraint-result}\,\, (f \, {\rm is} \, {\rm monotone}),\\
f(\bo) - 2f(\bq^w) &\le& f(\bar{\bo}^w) \label{without-constraint-result-non}\,\, ( f \, {\rm is} \, \text{non-monotone}).
\eea

Suppose that $\bs$ is a feasible solution appearing in Algorithm~\ref{greedy} containing $\bq^w$ as $\bs^0$.
If we can guarantee the quality of this solution, then the returned solution of Algorithm~\ref{greedy} will not be worse.
We know that one item $e\in E$ should be found in Line 5 during each iteration.
If $e\notin P(\bs)\cup P(\bo)$, we can omit $e$ from $E$, which does not affect the quality of the solution of Algorithm~\ref{greedy}. Then each $e$ selected in Line 5 but not added to $P(\bs)$ (because of the budget) should belong to $P(\bo)$. Suppose $p+1$ is the first iteration in which the element $e^{p+1}\in P(\bo)$ is considered but not added to $P(\bs)$.

Let $\bs^0=\bq^w$ be the starting solution of Algorithm~\ref{greedy} and $\bo^0=\bar{\bo}^w \setminus \bI_{[e^{p+1}, i^{p+1}_*]}$, where $i^{p+1}_*$ is the index of $e^{p+1}$ in $\bar{\bo}^w$.
For $j=1,2,\ldots,p$, define
\bea
\bs^j&=&\bs^{j-1}\sqcup \bI_{[e^j, i^j]}, \nn\\
\bo^j&=& ( \bo^{j-1}\sqcup \bI_{[e^j, i^j]} ) \sqcup \bI_{[e^j, i^j]}. \nn
\eea
That is, $\bs^j$ is the partial greedy solution after the $j$-th iteration, and $\bo^j$ consists of the items in $P(\bo^0)\cup P(\bs^j)$, where the index of items in $P(\bs^j)$ aligns with $\bs^j$, and the index of other items aligns with $\bo^0$. By denoting
\[\bo^{j-\frac{1}{2}}=\bo^{j-1}\sqcup \bI_{[e^j, i^j]},
\]we have the following relationship:
\bea
\bs^{j-1}&\preceq&\bo^{j-\frac{1}{2}},\label{patial1}\\
\bs^{j-1}&\preceq&\bo^{j-1}\label{patial2}.
\eea
Furthermore, we can see the construction process of sequences $\{\bs^j\}_{j=0}^{p}$, $\{\bo^j\}_{j=0}^p$ and $\{\bo^{j-\frac{1}{2}}\}_{j=1}^{p}$ in Figure~\ref{fig:Figure_2}.

\begin{figure}[tb]
    \centering
    \includegraphics[width=13cm,height=11cm,angle=0]{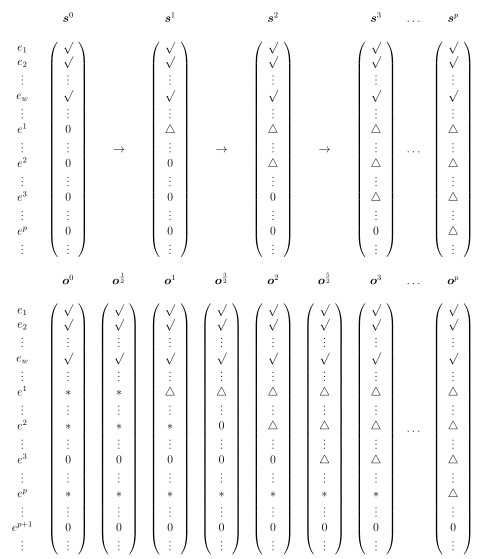}
    \caption{The constructing sequences of $\{\bs^j\}_{j=0}^{p}$, $\{\bo^j\}_{j=0}^p$ and $\{\bo^{j-\frac{1}{2}}\}_{j=1}^{p}$. The first three elements greedily selected by Algorithm~\ref{greedy} correspond to different cases, which are $\bo^0_{e^1} = \bs^1_{e^1} \ne 0$, $0 \ne \bo^0_{e^2} \ne \bs^2_{e^2} \ne 0$ and $\bo^0_{e^3} = 0 \ne \bs^3_{e^3}$. In particular, $\bo^0_{e^{p+1}} = 0$. }
    \label{fig:Figure_2}
\end{figure}

Moreover, for $t=0, \dots, p$, we also have
\bea\label{(supp-t)}
P(\bo^0)\setminus P(\bs^t) = P(\bo^t) \setminus P(\bs^t ).
\eea

Then we have the following lemma.
\begin{lemma}\label{lem123}
For the item $e^{p+1}\in P(\bo)\setminus P(\bs^{p})$ and $j\in [k]$, $\Delta_{e^{p+1}, j} f(\bo^0) \le \frac{1}{w}f(\bs^0)$.
\end{lemma}
\begin{proof}
By the orthant submodularity of $f$, the order rule of $P(\bo)$, the fact that $f(\varnothing)\geq 0$ and $\bq^w=\bs^0$, we can give the proof as follows:
\bea
\Delta_{e^{p+1} , j} f( \bo^0 )
&\le& \min\{f( \bq^{w-1} \sqcup \bI_{[e^{p+1}} ,j] ) - f( \bq^{w-1} ), \dots, f(\bq^0 \sqcup \bI_{[e^{p+1} ,j]})-f( \bq^0 ) \} \nn\\
&\leq& \min\{f( \bq^w )-f(\bq^{w-1}), \dots, f( \bq^1 )-f(\bq^0 ) \} \nn\\
&\leq& \frac{f( \bq^w ) - f( \bq^{w-1} ) + \dots + f( \bq^1 ) - f(\bq^0)}{w} \leq \frac{1}{w} f( \bq^w ).\nn
\eea
\endzm
\end{proof}

By analyzing each iteration before the $p$-th iteration, we can get the following lemma.
\begin{lemma}\label{lem1}
For $t=0, \dots, p$,

{\rm (I)} if $f$ is monotone,
\[
f(\bo^0)\le 2f(\bs^t)-f(\bs^0)+\sum_{e\in P(\bo^0)\setminus P(\bs^t)}\Delta_{e, \bo^0_e}f(\bs^t);
\]

{\rm (II)} if $f$ is non-monotone,
\[
f(\bo^0)\le3f(\bs^t)-2f(\bs^0)+\sum_{e\in P(\bo^0)\setminus P(\bs^t)}\Delta_{e, \bo^0_e}f(\bs^t). \]
\end{lemma}

\begin{proof}
The proof is trivial for $t=0$, so we only consider the cases $t=1,2,\ldots,p$.

(I) According to whether the element $e^j$ belongs to $P(\bo^{j-1})$ or not,
we first prove that
\be\label{main}
f(\bo^{j-1})-f(\bo^j) \le f(\bs^j)-f(\bs^{j-1}),  \,\,\forall j\in [p].
\ee

\textbf{Case I.1. If 
$e^j\in P(\bo^{j-1})$.} 

In this case, $\bo^{j-1}_{e^j}\ne0$. If $\bo^{j-1}_{e^j}=i^j$, then by the monotonicity of $f$, we have
$
f(\bo^{j-1})-f(\bo^j) = 0 \le f(\bs^j)-f(\bs^{j-1}).
$
Otherwise, we get
\bea
f(\bo^{j-1})-f(\bo^j) &= & f(\bo^{j-1})-f(\bo^{j-\frac{1}{2}})-(f(\bo^{j})-f(\bo^{j-\frac{1}{2}}))\nn \\
&\leq& f(\bo^{j-1})-f(\bo^{j-\frac{1}{2}}) \,\,({\rm monotonicity,\, since}\,\bo^{j-\frac{1}{2}}\preceq\bo^{j}) \nn \\
&= & \Delta_{e^j,\bo^{j-1}_{e^j}} f(\bo^{j-\frac{1}{2}})\,\,{\rm (definitions\, of\,}\bo^{j-\frac{1}{2}},\,\bo^{j-1})\label{case1-5} \\
&\leq & \Delta_{e^j,\bo^{j-1}_{e^j}} f(\bs^{j-1}) \,\,({\rm orthant\, submodularity\,\, and}\,\,(\ref{patial1})) \label{case1-6}\\
&\leq & f(\bs^j)-f(\bs^{j-1}) \,\,{\rm (Algorithm~\ref{greedy})}\label{case1-7}.
\eea

\textbf{Case I.2. If 
$e^j\notin P(\bo^{j-1})$.} 

In this case, $\bo^{j-1}_{e^j}=0$ and $\bo^{j}=\bo^{j-\frac{1}{2}}$. Then we have
\[
f(\bo^{j-1})-f(\bo^j)= -\Delta_{e^j, i^j} f(\bo^{j-1}) \le 0\le f(\bs^j)-f(\bs^{j-1}).
\]

Thus, we can finish the proof of (\ref{main}). For $t=1, \dots, p$, by summing the first $t$-terms of these inequalities, we have
\begin{equation}\label{(2)}
f(\bo^0)-f(\bo^t) \le f(\bs^t)-f(\bs^0).
\end{equation}
By Lemma~\ref{lem-ksub} and $\bs^t\preceq \bo^t$, we can obtain
\begin{equation}\label{(3)}
f(\bo^t)\le f(\bs^t)+\sum_{e\in P(\bo^t)\setminus P(\bs^t)}\Delta_{e, \bo^t_e}f(\bs^t), t=1, \dots, p.
\end{equation}
Thus, we can finish the proof of (I) by combining (\ref{(supp-t)}), (\ref{(2)}) and (\ref{(3)}).

(II) In this part, assume that $f$ is non-monotone. Following the proof similar to that of (I), for $j\in [p]$, if the following inequality is correct
\begin{equation}\label{supp}
f(\bo^{j-1})-f(\bo^j) \le  2[f(\bs^j)-f(\bs^{j-1})],
\end{equation}
then by accumulating $j$ from $1$ to $t$, we can get
\[
f(\bo^0)-f(\bo^t)\le2[f(\bs^t)-f(\bs^0)], \, t=1, \dots, p.
\]
Therefore, by Lemma~\ref{lem-ksub} and~(\ref{(supp-t)}), we finish our proof.

\textbf{Case II.1. If 
$e^j\in P(\bo^{j-1})$.} 

When $\bo^{j-1}_{e^j}=i^j$, the left hand side of Inequality~(\ref{supp}) is $0$. Then we will explain that its right hand side is non-negative. By Lemma~\ref{mon-ksub}, we know that $f$ is pairwise monotone,
then there is at most one position $i\in [k]$ satisfying $\Delta_{e^j,i}f(\bs^{j-1})<0$. In fact, if this case occurs, it must be the position with the smallest marginal benefit.
Thus, based on the greedy technique used in Step $5$ of Algorithm~\ref{greedy}, even if $f$ is non-monotone, the right hand side of Inequality~(\ref{supp}) could not be negative.

Otherwise ($\bo^{j-1}_{e^j}\neq i^j$), also from the pairwise monotonicity of $f$,
we have
\bea
f(\bo^{j-1})-f(\bo^j)& = & 2[f(\bo^{j-1})-f(\bo^{j-\frac{1}{2}})]-[f(\bo^{j-1})+f(\bo^j)-2f(\bo^{j-\frac{1}{2}})]\nn\\
& \le & 2[f(\bo^{j-1})-f(\bo^{j-\frac{1}{2}})].\nn
\eea
Now Inequality~(\ref{supp}) can be obtained by using the same arguments as those in (\ref{case1-5})-(\ref{case1-7}).

\textbf{Case II.2. If 
$e^j\notin P(\bo^{j-1})$.} 

In this case, $\bo^{j-1}_{e^j}=0$ and $\bo^{j}=\bo^{j-\frac{1}{2}}$. Then we have
\bea
f(\bo^{j-1})-f(\bo^j) &=&-\Delta_{e^j,i^j} f(\bo^{j-1})\nn\\
&\leq& \Delta_{e^j,i'} f(\bo^{j-1}),\,\,\forall i'\in[k]\setminus\{i^j\} \,\,{\rm (pairwise\, monotone)}\nn \\
&\leq& \Delta_{e^j,i'} f(\bs^{j-1})\,\,({\rm orthant\,\, submodular\,\, and}\,\,(\ref{patial2}))\nn\\
& \leq & \Delta_{e^j,i^j} f(\bs^{j-1})\,\,({\rm Algorithm}~\ref{greedy})\nn\\
& = & f(\bs^j)-f(\bs^{j-1}).\nn
\eea
From the explanation on the application of the pairwise monotonicity of a $k$-submodular function and the greedy technique used in Step $5$ of  Algorithm~\ref{greedy}, we show that $f(\bs^j)-f(\bs^{j-1})\geq 0$. Therefore, we know that Inequality~(\ref{supp}) holds.
\endzm
\end{proof}
\begin{theorem}\label{the1}
If $f$ is monotone, then Algorithm \ref{greedy} returns a $\frac{1}{2}(1-e^{-2})$ approximation solution with query complexity $O(n^6 k^5)$.
\end{theorem}

\begin{proof}
First, the query complexity is $O(n^{6}k^{5})$ because we need to enumerate four elements in any position and loop $n$ times, where we calculate the marginal density by querying $n k$ times the values of $f$.

Next we analyze the approximation ratio.
For $j\in [p+1]$, by defining $\theta_j=\max_{e\in E^{j-1}, i\in[k]}\frac{\Delta_{e, i}f(\bs^{j-1})}{c_e}$, we have $\Delta_{e, \bo^0_e} f(\bs^{j-1}) \le c_e\theta_{j}$ for each $e\in P(\bo^0)\setminus P(\bs^{j-1})$.
Then for $t=0, \dots, p$, $\sum_{e\in P(\bo^0)\setminus P(\bs^t)} \Delta_{e, \bo^0_e}f(\bs^t) \leq (c(P(\bo^0) \setminus P(\bs^0))) \theta_{t+1}$, since $c(P(\bo^0)\setminus P(\bs^t))\le c(P(\bo^0) \setminus P(\bs^0))$.
Now define a new function $g(\bx)=f(\bx \sqcup \bs^0)-f(\bs^0)$ for all $\bx \in (k+1)^{V\setminus P(\bs^0)}$, which satisfies $k$-submodularity.
Together with Lemma \ref{lem1} (I), we get
\bea\label{(5)}
g(\bo^0 \setminus \bs^0)
&\leq& 2g(\bs^t \setminus \bs^0)+\sum_{e\in P(\bo^0)\setminus P(\bs^t)}\Delta_{e, \bo^0_e}f(\bs^t) \nn\\
&\leq& 2\big[g(\bs^t \setminus \bs^0)+\frac{c(P(\bo^0) \setminus P(\bs^0))}{2}\theta_{t+1}\big].
\eea


Let $L_t=\sum_{j=1}^t c_{e^j}$ and $L_0=0$.
Define $L' =  L_p = c(P(\bs^p) \setminus P(\bs^0))$ and $L'' = c(P(\bo^0) \setminus P(\bs^0))$.
It is easy to see that $L' \ge L''$, as otherwise we have $L' + c(e^{p+1}) < L'' + c(e^{p+1}) \le L - c(P(\bs^0))$, giving a contradiction to the definition of $p+1$.
For $l=1, \dots, L_p$, we define new variables $\rho_l=\theta_{t}$, if $l=L_{t-1}+1, \dots, L_t$, i.e., $\rho_1 = \cdots = \rho_{L_1} = \theta_1, \rho_{L_1+1} = \cdots = \rho_{L_2} = \theta_2, \ldots, \rho_{L_{p-1} + 1= \cdots = \rho_{L_p} = \theta_p}$. Then by collecting items with the same variables, we can obtain
\be\label{min}
\min_{s=1, \dots, L_p } \{ \sum_{l=1}^{s-1} \rho_l + \frac{L''} {2} \rho_s \} = \min_{t=0, \dots, p-1} \{ \sum_{l=1}^{L_t} \rho_l + \frac{L''}{2} \rho_{L_t+1}\}.
\ee

Furthermore, from this definition, we also have
$g(\bs^t)=\sum_{\tau=1}^tc_{e^{\tau}}\theta_{\tau}=\sum_{l=1}^{L_t}\rho_l$ for $t=1, \dots, p$.
Thus, we get
\[
\min_{t=0, \dots, p-1}\{g(\bs^t \setminus \bs^0)+\frac{L''}{2} \theta_{t+1}\}=\min_{t=0, \dots, p-1 }\{\sum_{l=1}^{L_t}\rho_l+\frac{L''}{2}\rho_{L_t+1}\}.
\]
Then together with~(\ref{min}), (\ref{(5)}) can be improved as follows.
\bea\label{(6)}
g(\bo^0 \setminus \bs^0)
&\le&  2\min_{t=0, \dots, p-1}\{g(\bs^t \setminus \bs^0)+\frac{L''}{2}\theta_{t+1}\} \nn\\
&=&  2\min_{s=1, \dots, L_p}\{\sum_{l=1}^{s-1}\rho_l+\frac{ L'' }{2}\rho_s\}.
\eea
By Lemma \ref{lemA111} and (\ref{(6)}), we obtain
\begin{align}\label{(7)}
\frac{g(\bs^{p} \setminus \bs^0) } {g(\bo^0\setminus \bs^0)} \notag
= \frac{\sum_{l=1}^{L_p }\rho_l}{g(\bo^0\setminus \bs^0)}
\ge{} & \frac{\sum_{l=1}^{L_p} \rho_l}{2\min_{s=1, \dots, L_p} \{\sum_{l=1}^{s-1}\rho_l+\frac{L'' }{2}\rho_s\}} \\
\ge{} & \frac{1}{2}(1-e^{-2L' / L''}) \geq \frac{1}{2}(1-e^{-2}).
\end{align}

Moreover, by (\ref{without-constraint-result}) and Lemma~\ref{lem123}, we have
\bea
\frac{w+1}{w} f(\bs^0) \notag
&=& f(\bs^0) + \frac{1}{w} f(\bs^0) \nn\\
&\ge& f(\bs^0) +  \Delta_{e^{p+1}, i^{p+1}_*} f(\bo^0) \nn\\
&=& f(\bar{\bo}^w) - [f(\bo^0) - f(\bs^0)] \nn\\
&\ge& f(\bo) - f(\bs^0) - g(\bo^0 \setminus \bs^0), \nn
\eea
which implies that
\be\label{s0-monotone}
\frac{2w+1}{w} f(\bs^0) \ge f(\bo) - g(\bo^0 \setminus \bs^0).
\ee

Finally, combining (\ref{(7)}) and (\ref{s0-monotone}), we obtain a lower bound on the output $f(\bs)$ of out algorithm:
\bea
f(\bs) \ge f(\bs^p) &=& f(\bs^0) + g(\bs^p \setminus \bs^0) \nn\\
&\ge& f(\bs^0) + \frac{1}{2} (1-e^{-2}) g(\bo^0 \setminus \bs^0) \nn\\
&\ge& \frac{w}{2w+1} ( f(\bo) - g(\bo^0 \setminus \bs^0) ) + \frac{1}{2} (1-e^{-2}) g(\bo^0 \setminus \bs^0) \nn\\
&\ge& \min\{ \frac{w}{2w+1}, \frac{1}{2} (1-e^{-2}) \} \cdot f(\bo). \nn
\eea
When $w\ge 4$, $\frac{w}{2w+1} \ge \frac{1}{2} (1-e^{-2})$, and thus $f(\bs) \ge \frac{1}{2} (1-e^{-2}) f(\bo)$.

\endzm
\end{proof}


According to the similar process as that in Theorem~\ref{the1}, combined with inequalities (\ref{without-constraint-result-non}) and Lemma~\ref{lem123}, we present the  following result for non-monotone case.
\begin{theorem}\label{non-montone}
If $f$ is non-monotone, then Algorithm \ref{greedy} has an approximation ratio of $\frac{1}{3}(1-e^{-3})$ with query complexity $O(n^{9}k^{8})$.
\end{theorem}

\end{document}